%% file: main.tex
\documentclass[11pt]{article}
\usepackage{fullpage}
\usepackage{amsmath,amsfonts,amsthm,amssymb,bbm}
\usepackage[capitalise]{cleveref}
\usepackage{tikz}
\usepackage{xcolor}
\newcommand{\cc}{\mathbb{C}}
\newcommand{\id}{\mathbf{I}}
\newcommand{\un}{\mathcal{U}}
\newcommand{\dm}{\mathcal{D}}
\newcommand{\ud}{\Phi_\mathrm{D}}
\newcommand{\ul}{\Phi_\mathrm{L}}
\newcommand{\e}{\mathop{\mathbf{E}}\limits}
\newcommand{\ket}[1]{|#1\rangle}
\newcommand{\bra}[1]{\langle#1|}

\newcommand{\Tr}{\mathrm{Tr}}

\newcommand{\ct}{^\dagger}
\newcommand{\mt}{^{\mathstrut}}
\newcommand{\vect}{\mathrm{vec}}
\newcommand{\poly}{\mathrm{poly}}
\newcommand{\polylog}{\mathrm{polylog}}
\newcommand{\diag}{\mathrm{diag}}

\newtheorem{theorem}{Theorem}[section]
\newtheorem{lemma}[theorem]{Lemma}

\newtheorem{corollary}[theorem]{Corollary}
\newtheorem*{claim}{Claim}
\theoremstyle{remark}
\newtheorem*{remark}{Remark}
\theoremstyle{definition}
\newtheorem{definition}[theorem]{Definition}

\title{Quantum Logspace Algorithm for Powering Matrices with Bounded Norm}
\date{}
\author{Uma Girish\thanks{Department of Computer Science, Princeton University. Research supported by the Simons Collaboration on Algorithms and Geometry, by a Simons Investigator Award and by the National Science Foundation grant No. CCF-1714779.}
\and Ran Raz\thanks{Department of Computer Science, Princeton University. Research supported by the Simons Collaboration on Algorithms and Geometry, by a Simons Investigator Award and by the National Science Foundation grant No. CCF-1714779.}
\and Wei Zhan\thanks{Department of Computer Science, Princeton University. Research supported by the Simons Collaboration on Algorithms and Geometry, by a Simons Investigator Award and by the National Science Foundation grant No. CCF-1714779.}}

\begin{document}

\maketitle
	
\begin{abstract}
We give a quantum logspace algorithm for powering contraction matrices, that is, matrices with spectral norm at most~1. The algorithm gets as an input an arbitrary $n\times n$ contraction matrix $A$, and a parameter $T \leq \poly(n)$ and outputs the entries of $A^T$, up to (arbitrary) polynomially small additive error. The algorithm applies only unitary operators, without intermediate measurements. We show various implications and applications of this result:

First, we use this algorithm to show that the class of quantum logspace algorithms with only quantum memory and with intermediate measurements is equivalent to the class of quantum logspace algorithms with only quantum memory without intermediate measurements. This shows that the deferred-measurement principle, a fundamental principle of quantum computing, applies also for quantum logspace algorithms (without classical memory). More generally, we give a
quantum algorithm with space $O(S + \log T)$ that takes as an input the description of a quantum algorithm with quantum space $S$ and time $T$, with intermediate measurements (without classical memory), and simulates it unitarily with polynomially small error, without intermediate measurements.

Since unitary transformations are reversible (while measurements are irreversible) an interesting aspect of this result is that it shows that any quantum logspace algorithm (without classical memory) can be simulated by a reversible quantum logspace algorithm. This proves a quantum analogue of the result of Lange, McKenzie and Tapp that deterministic logspace is equal to reversible logspace~\cite{lange2000reversible}.


Finally, we use our results to show non-trivial classical simulations of quantum logspace learning algorithms.

	\end{abstract}
	
	\input{introduction}
	\input{preliminaries}
	\input{implement}
	\input{unital}
	\input{powering}
	\input{learning}
	
	\section*{Acknowledgement}
	We would like to thank Dieter van Melkebeek and Subhayan Roy Moulik for very helpful suggestions and comments on a previous version of this work. We also thank the anonymous reviewers for their thorough feedback.
	
	\bibliography{main}
	\bibliographystyle{alpha}
	
	\input{appendix}

\end{document}

%% file: introduction.tex
\section{Introduction}

Quantum computers hold great promise, but in the near future their memory is likely to be limited to a small number of qubits. This motivates the study of quantum complexity classes with bounded space. The most important of these classes is the class of problems solvable in quantum logarithmic space and polynomial time, first studied by Watrous~\cite{watrous1999space}. In the literature, there are several variants of this class. One variant, $\mathsf{BQL}$, is the class of problems solvable in quantum logarithmic space and polynomial time when intermediate measurements are allowed. Another variant, $\mathsf{BQ_UL}$, is the class of problems solvable in quantum logarithmic space and polynomial time when only unitary operators are allowed and intermediate measurements are not allowed. 
We note that in most previous works, the class $\mathsf{BQL}$ allows a quantum algorithm to use both quantum and classical memory
(see for example~\cite{melkebeek2012time,ta2013inverting,fefferman2018complete}).

Our first main result, \cref{thm:general}, gives a  quantum logspace algorithm for powering matrices, a fundamental problem in computational complexity, which is not known to be in classical (deterministic or probabilistic) logspace. Our algorithm uses only unitary operators, without  intermediate measurements, and hence it places the problem of powering matrices in the class $\mathsf{BQ_UL}$.

The algorithm gets as an input an arbitrary $n\times n$ matrix $A$, a parameter $T \leq \poly(n)$ and two indices $i,j \in \{1,\ldots,n\}$  and outputs the entry  $(A^T)_{i,j}$, up to an additive error of $\tfrac{\|A\|^T}{\poly(n)} + \tfrac{1}{\poly(n)} $, where $\|A\|$ is the spectral norm of the matrix $A$. In particular, if $A$ is a  contraction matrix, that is, a matrix with spectral norm at most~1, the additive error is just $\tfrac{1}{\poly(n)}$.

We note that by an easy reduction, our algorithm can also solve another fundamental problem in computational complexity, the problem of iterative matrix multiplication. In this problem, the input is $T$  matrices $A_1,\ldots,A_T$ of size $n\times n$ each, and the algorithm outputs the entries of the product $A_1 \cdot \ldots \cdot A_T$.

Besides giving a quantum logspace algorithm for a basic computational problem, our results shed light on several fundamental issues regarding bounded-space quantum computations, and have additional applications.

\subsubsection*{$\mathsf{BQ_QL}$ is Equal to $\mathsf{BQ_UL}$}

We consider the class of quantum logspace algorithms with only quantum memory and with intermediate measurements and refer to it by $\mathsf{BQ_QL}$.
We use our algorithm for powering matrices to show that the two classes $\mathsf{BQ_QL}$ and $\mathsf{BQ_UL}$ are exactly equal (\cref{thm:main}). Moreover, the way that this equality is proved is by a {\it simulation}. Our second main result, \cref{thm:unitals}, proves that there is a quantum logspace algorithm without intermediate measurements, that is, a $\mathsf{BQ_UL}$ algorithm, that gets the description of a quantum logspace algorithm with intermediate measurements, without classical memory, that is, a $\mathsf{BQ_QL}$ algorithm, and simulates it with polynomially small error. \cref{thm:unitals} is even more general and shows how to simulate a quantum logspace algorithm with {\it unital} channels that are given as an input, while even the very restricted special case of simulating an arbitrary unitary operator within $\mathsf{BQ_UL}$ seems to us interesting.

\subsubsection*{The Deferred-Measurement Principle}

The deferred measurement principle is a fundamental result in quantum computing which states that delaying measurements until the end of a computation doesn't affect the output. In order for the principle to hold, the qubits that were supposed to be measured cannot further participate in the computation from that point on. However, a $\mathsf{BQ_QL}$ algorithm can only store a logarithmic number of qubits, while the number of intermediate measurements is potentially polynomial, and hence excluding the qubits that are supposed to be measured from further participating in the computation is infeasible.

Nevertheless, Theorem~\ref{thm:main} and Theorem~\ref{thm:unitals} imply
that intermediate measurements are not necessary even when the space used by the quantum algorithm is logarithmic, but the way to eliminate the intermediate measurements is not as straightforward.

\subsubsection*{Reversible Computation}

Landauer introduced the concept of time-reversible computation and argued that any irreversible operation must be accompanied by entropy increase~\cite{Landauer} (see also~\cite{bennett}).
An interesting aspect of \cref{thm:main} and \cref{thm:unitals} is that they show that any quantum logspace algorithm (without classical memory) can be implemented using only time-reversible operations (except for the final measurement that gives the final output). This is a quantum analogue of the result of Lange, McKenzie and Tapp that deterministic logspace algorithms can be implemented using only time-reversible operations~\cite{lange2000reversible}.


\subsubsection*{Classical Simulations of Quantum Learning with Bounded Memory}

A line of recent works studied the power of (classical) algorithms for online learning, under memory constraints, where a bounded-space learner tries to learn a concept class from a stream of samples. These works showed that for a large class of online learning problems, any classical learning algorithm requires either super-linear memory size or a super-polynomial number of samples (see for example~\cite{S14,SVW16,R16,KRT17,R17,MM18,BOGY18,GRT18} and the references therein).

Here, we study the relative power of quantum and classical algorithms for online learning, under memory constraints. More concretely, we study the task of distinguishing between two families of distributions over the possible samples.
\cref{col:NC2} proves that any quantum algorithm with time $T$ and space $S$ for distinguishing between arbitrary two families of distributions, can be simulated classically in time $\poly(2^{S^2 + \log^2 T})$ and space $O(S^2 + \log^2 T)$. Moreover, \cref{thm:singleton} proves that if one family is a singleton, that is, the task is to distinguish between one distribution over the samples and a family of different distributions, then any quantum learning algorithm with time $T$ and space $S$  can be simulated classically in time $\poly(2^S \cdot T)$ and space $O(S + \log T)$.

Thus, an intriguing open problem is whether any quantum algorithm with time $T$ and space $S$ for distinguishing between two arbitrary families of distributions,
can be simulated classically in time $\poly(2^S \cdot T)$ and space $O(S + \log T)$. \cref{thm:lequiv} proves that this holds if and only if $\mathsf{promiseBQ_UL}=\mathsf{promiseBPL}$.

\subsection{Techniques}

We start by proving a lemma that shows how to implement an arbitrary contraction matrix $A$ as a subsystem of a unitary quantum circuit (\cref{thm:implem}). Since $A$ is not necessarily unitary, rather than implementing $A$, the lemma implements the unitary matrix
	\[U_H=\begin{pmatrix}
		H & \sqrt{\id_{2m}-H^2} \\
		\sqrt{\id_{2m}-H^2} & -H
	\end{pmatrix}\]
where $H$ is the Hermitian contraction
\[\begin{pmatrix} & A \\ A\ct & \end{pmatrix}.\]
That is, the lemma shows how to apply the transformation $U_H$ on a unit vector (quantum state) that is also given as an input. The unitary matrix $U_H$ is called a block-encoding of $A$ in some literature \cite{chakraborty2019power,gilyen2019quantum2}, which admits various constructions (see \cite{gilyen2019quantum} for an exhibition). In particular, our construction in \cref{thm:implem} is in unitary quantum logspace.

The proof of \cref{thm:implem} is inspired by, and uses techniques from, Ta-Shma's algorithm that inverts well-conditioned matrices in quantum logspace~\cite{ta2013inverting}, whose general framework traces back to \cite{harrow2009quantum}. In particular, as in~\cite{ta2013inverting}, the proof goes according to the following lines: Given a Hermitian matrix $H$,
\begin{itemize}
	\item First apply the phase estimation over the unitary $e^{iH}$ so that it maps $\ket{u_\lambda}$ to $\ket{u_\lambda}\ket{\lambda}$, where $u_\lambda$ is an eigenvector of $H$ with eigenvalue $\lambda$.
	\item For each eigenvector apply the unitary transformation $\ket{\lambda}\rightarrow \lambda\ket{0}\ket{\lambda}+\sqrt{1-\lambda^2}\ket{1}\ket{\lambda}$ according to the eigenvalue $\lambda$. This is where contraction matrices come into play, as the eigenvalues of $H$ are required to be in $[-1,1]$.
	\item Uncompute the eigenvalues by reversing the phase estimation over $e^{iH}$.
\end{itemize}

As a special case  of \cref{thm:implem}, when we take the contraction $A$  to be unitary, we get a space-efficient unitary implementation of any unitary matrix (\cref{col:implem}).

We get our algorithms for powering contraction matrices (\cref{thm:power} and \cref{col:power}) by iteratively applying the unitary matrix $U_H$ of \cref{thm:implem}. However, since \cref{thm:implem} implements the matrix $U_H$, rather than $A$, we need to `throw away' the unwanted dimensions introduced by $U_H$, by permuting them into additional dimensions.

We get our algorithm for powering arbitrary matrices (\cref{thm:general}), by a reduction to the algorithm for powering contraction matrices, by dividing the matrix by its norm. However, the known algorithm for computing the spectral norm of a matrix, by Ta-Shma~\cite{ta2013inverting}, only works for contraction matrices. To bypass this, we apply Ta-Shma's algorithm on the matrix $A$ divided by its Frobenius norm (which is always larger than the spectral norm).

Finally, we get our algorithms for simulating quantum logspace algorithms with intermediate measurements, or even {\it unital} channels that are given as an input (\cref{thm:unital} and \cref{thm:unitals}), by reducing any unital quantum algorithm to the contraction powering problem in the $m^2$-dimensional space of the $m \times m$ entries of the density matrix, where $m=2^S$ and $S$ is the space used by the algorithm. Note that this step already doubles the space used. At the end of this step, we only get polynomially small success probability, but that success probability can be amplified to a constant using a Grover-type technique inspired by \cite{fefferman2018complete}, resulting in \cref{thm:unital} that simulates the computation with constant error. The error is further reduced to be polynomially small in Theorem~\ref{thm:unitals}. Interestingly, to reduce the error and prove \cref{thm:unitals}, we use \cref{thm:main}, so, in a way, the results are used to improve themselves.

\subsection{Related Work}
Independently of our work, Fefferman and Remscrim have proven closely related results to ours~\cite{fefferman2020eliminating}. They took a different route from ours by proving $\mathsf{L}$-reductions between several well-conditioned versions of matrix problems which turned out to be $\mathsf{BQ_UL}$-complete. In particular, they obtained a stronger version of our \cref{thm:main}, showing that $\mathsf{BQL} = \mathsf{BQ_UL}$.

\subsection{Paper Organization}
The remainder of the paper is organized as follows. In Section~\ref{sect:pre}, we review the definitions of contractions, quantum channels and quantum algorithms, and provide useful tools for the proofs later on. In Section~\ref{sect:impl}, we show how to implement an arbitrary contraction as a subsystem of a unitary quantum circuit, using the same algorithm framework as \cite{harrow2009quantum}.
As a by-product, we show how to implement an arbitrary unitary matrix given the description of the matrix.
In Section~\ref{sect:cpowering}, we give our first algorithm for contraction powering.
In Section~\ref{sect:unital}, we prove \cref{thm:main} by reducing any unital quantum algorithm to the contraction powering problem, and showing the later can be computed in unitary quantum logspace. In Section~\ref{sect:powering}, we show how to power general square matrices in unitary quantum logspace via our contraction powering algorithm. In Section~\ref{sect:learn}, we discuss the relations between contraction powering and quantum learning algorithms. 

%% file: preliminaries.tex
\section{Preliminaries}\label{sect:pre}
	
	For an integer $n$, let $[n]=\{0,1,\ldots,n-1\}$. Let $\cc$ denote the set of complex numbers, and $\cc^{m\times n}$ denote the set of $m$ by $n$ complex matrices. For a matrix $A\in\cc^{m\times n}$, let $\vect(A)$ be the vectorization of $A$, which is a vector of dimension $mn$ formed by stacking the columns of $A$ on top of each other, that is
	\[\vect(A)_{i+jm}\mt=A_{i,j}\mt,\quad\forall i\in[m],j\in[n].\]
	Let $\un_m$ be the set of $m$ by $m$ unitary matrices, and $\dm_m$ be the set of $m$ by $m$ density matrices, i.e. positive semidefinite Hermitians of trace $1$. The $m$ by $m$ identity matrix is denoted by $\id_m$. Let $\|A\|$ denote the spectral norm of a complex matrix $A$, and $\|A\|_\mathrm{F}$ denote the Frobenius norm.
	
	We use $\varepsilon$ to denote small real numbers, and $\ket{\epsilon}$ to denote vectors with small norms. When we talk about errors, approximations and $\varepsilon$-closeness of matrices, they are measured in spectral norms.
	
	As we work mostly with complex numbers, we often need corresponding concentration bounds. The following is a direct corollary of the Chernoff-Hoeffding inequality:
	\begin{lemma}\label{lemma:ch}
		Let $X$ be a random complex number with $|X|\leq 1$, and $X_1,\ldots,X_n$ are $n$ independent copies of $X$. Then
		\[\Pr\left[\left|\frac{1}{n}(X_1+\cdots+X_n)-\e[X]\right|\geq\varepsilon\right]\leq 4e^{-2n\varepsilon^2}.\]
	\end{lemma}
	
	\subsection{Contraction Matrices}
	
	We introduce contraction matrices and provide some useful properties, which can be found in \cite[Chapter 6]{zhang2011matrix}:
	\begin{definition}
		A matrix $A\in\cc^{m\times m}$ is a \textit{contraction} if $\|A\|\leq 1$. Alternatively, $A$ is a \textit{contraction} if $A$ is in the convex hull of $\un_m$.
	\end{definition}
	Any eigenvalue $\lambda$ of a contraction must have $|\lambda|\leq 1$. If $A\in\cc^{m\times m}$ is a contraction, then the following matrix is unitary:
	\[U_A=\begin{pmatrix}
	A & \sqrt{\id_m-AA\ct} \\
	\sqrt{\id_m-A\ct A} & -A\ct
	\end{pmatrix}\]
	In particular, when $A=a$ is a real number in $[-1,1]$, $U_a=\begin{pmatrix} a & \sqrt{1-a^2} \\ \sqrt{1-a^2} & -a\end{pmatrix}$ is a reflection. Finally, in a product of multiple contractions, individual errors will not propagate much, as we have the following lemma:
	\begin{lemma}\label{lemma:error}
		If $A_1,\ldots,A_k\in\cc^{m\times m}$ are contractions, and $B_1,\ldots,B_k\in\cc^{m\times m}$ satisfy $\|A_i-B_i\|\leq\varepsilon$ for every $i$, then $\|A_1\cdots A_k-B_1\cdots B_k\|\leq (1+\varepsilon)^k-1$. Furthermore, if $B_1,\ldots,B_k$ are also contractions, then $\|A_1\cdots A_k-B_1\cdots B_k\|\leq k\varepsilon$.
	\end{lemma}
	\begin{proof}
		Use induction on $k$. When $k=1$ the result holds by assumption. For $k\geq 2$, we have
		\begin{align*}
		\|A_1\cdots A_k-B_1\cdots B_k\| &\leq \|A_1\cdots A_{k-1}(A_k-B_k)\|
		+\|(A_1\cdots A_{k-1}-B_1\cdots B_{k-1})B_k\| \\
		&\leq \|A_k-B_k\|+\|A_1\cdots A_{k-1}-B_1\cdots B_{k-1}\|\cdot\|A_k-(A_k-B_k)\| \\
		&\leq \varepsilon+[(1+\varepsilon)^{k-1}-1](1+\varepsilon) \\
		& = (1+\varepsilon)^k-1.
		\end{align*}
		And when $B_1,\ldots,B_k$ are contractions,
		\begin{align*}
		\|A_1\cdots A_k-B_1\cdots B_k\|
		&\leq \|A_k-B_k\|+\|A_1\cdots A_{k-1}-B_1\cdots B_{k-1}\|\cdot\|B_k\| \\
		&\leq \varepsilon+(k-1)\varepsilon=k\varepsilon.\qedhere
		\end{align*}
	\end{proof}
	
	\subsection{Quantum Channels}
	
	A quantum channel (or operation), in its most general form, is a \textit{completely-positive trace-preserving} (CPTP) map $\Phi:\dm_m\rightarrow\dm_n$ that maps a density matrix $\rho$ to a density matrix $\Phi(\rho)$. We denote the set of such channels as $\mathcal{C}_{m,n}$. The \textit{Kraus representation} of the quantum channel $\Phi$ is a set of matrices $\{E_1,\ldots,E_k\}$ such that $\sum_{i=1}^k E_i\ct E_i\mt=\id_m$, and
	\[\Phi(\rho)=\sum_{i=1}^k E_i\mt\rho E_i\ct.\]
	The \textit{natural representation} of $\Phi$, denoted as $K(\Phi)$, is a matrix in $\cc^{n^2\times m^2}$ such that $\vect(\Phi(\rho))=K(\Phi)\vect(\rho)$ for any $\rho\in\dm_m$. Given the Kraus representation $\{E_1,\ldots,E_k\}$ of $\Phi$, one can easily compute the natural representation $K(\Phi)=\sum_{i=1}^k \overline{E_i}\otimes E_i$.
	
	By Stinespring's dilation theorem \cite{stinespring1955positive}, any quantum channel can be viewed as a unitary operation in a larger Hilbert space. In other words, any quantum channel can be decomposed into the composition of adding ancillas, performing a unitary and tracing out the ancillas. It was shown in \cite{iten2017quantum,shen2017quantum} how to construct arbitrary quantum channels in the quantum circuit model (discussed in detail below) with only $1$ additional ancilla qubit if intermediate measurements are allowed. However, it is not clear whether their construction can be computed in logspace.
	
	\paragraph{Unital channels}
	A quantum channel $\Phi$ is \textit{unital}, if it maps the identity to the identity of the same dimension. The  Kraus representation of a unital channel is a set of square matrices $\{E_1,\ldots,E_k\}$ that additionally satisfies $\sum_{i=1}^k E_i\mt E_i\ct=\mathbf{I}_m$. In the language of natural representation, it is known  that $\Phi$ is unital if and only if $K(\Phi)$ is a contraction \cite{perez2006contractivity}.
	
	A well-studied subclass of unital channels is the \textit{mixed-unitary} (or \textit{random-unitary}) channels. The Kraus representation of a mixed-unitary channel consists of matrices proportional to unitaries; equivalently, $\Phi$ is mixed-unitary if
	\[\Phi(\rho)=\sum_{i=1}^k p_i\mt U_i\mt\rho U_i\ct,\]
	where $\{p_i\}$ is a probability distribution and each $U_i$ is a unitary matrix. Mixed-unitary channels are clearly unital, but the converse is not true: there are unital channels which are not mixed-unitary \cite{audenaert2008random}, and it's recently shown that deciding whether a unital channel is mixed-unitary is \textsf{NP}-hard \cite{lee2019detecting}.
	
	Notice that unitary operators and projective measurements are all mixed-unitary (see e.g. \cite[Proposition 4.6]{watrous2018theory}), and therefore the above results imply that it is impossible to realize arbitrary unital channels in the quantum circuit model without adding ancillas, even with intermediate measurements. Rather surprisingly, our paper shows the other side of the fact: one can construct in logspace a quantum circuit to simulate any arbitrary unital channel with ancillas, but without intermediate measurements.
	
	\subsection{Quantum Algorithms}\label{sect:def}
	
	A generic quantum algorithm with time $T$ and space $S=\log m$ is specified by $T$ quantum channels $\Phi_1,\ldots,\Phi_T\in\mathcal{C}_{m,m}$, which might depend on the inputs. We also require the channels $\Phi_1,\ldots,\Phi_T$ to be efficiently constructible, whose meaning may differ for different types of quantum algorithms, and will be specified below.
	
	The algorithm starts from the fixed initial state $\rho_0\mt=\ket{0^S}\bra{0^S}$, and in the $i$-th step applies $\Phi_i$ on the current state, so that the state after the $i$-th step can be described as 
	\[\rho_i=\Phi_i(\rho_{i-1})=\Phi_i\circ \Phi_{i-1}\circ\cdots\circ \Phi_1(\rho_0).\]
	At the end the first qubit of the final state $\rho_T\mt$ is measured in the computational basis of the first qubit, where the measurement can be represented as $M_0=|0\rangle\langle0|\otimes \mathbf{I}_{m/2}$. The quantum algorithm outputs $0$ with probability $\Tr[\rho_T\mt M_0\mt]$, and $1$ with probability $1-\Tr[\rho_T\mt M_0\mt]$. The measurement $M_0$ is general enough, as every two-outcome measurement in the computational basis can be realized by permutations and $M_0$ with at most $1$ more ancilla qubit. The quantum algorithm is called unitary (resp. unital), if every channel $\Phi_i$ is unitary (resp. unital).
	
	\paragraph{Quantum circuit} Fix a universal quantum gate set $\mathcal{G}$, for instance Hadamard and Toffoli gates \cite{shi2003both}, and let $\mathcal{G}_S$ be the set of gates in $\mathcal{G}$ on $S$ qubits. Let $\mathcal{M}_S$ be the set of single-qubit measurements on $S$ qubits.
	
	When the input of the problem is from domain $X$, the quantum circuit is specified by a mapping $\ud:X\times[T]\rightarrow\mathcal{G}_S\cup\mathcal{M}_S$ such that $\Phi_{i+1}=\ud(x,i)$ for every $i\in[T]$, where $x\in X$ is the input, and $\ud$ can be computed deterministically in time $O(T)$ and space $O(S)$. The quantum algorithm decides a function $f:X\rightarrow\{0,1\}$ with error $\varepsilon$ if:
	\[ \forall x\in X,\quad |\Tr[\rho_T\mt M_0\mt]-f(x)|\leq\varepsilon. \]
	Now, $\mathsf{BQ_QL}$ is the class of boolean-function families where $f_n:\{0,1\}^n\rightarrow\{0,1\}$ can be decided by a quantum circuit with time $\poly(n)$, space $O(\log n)$ and error $1/3$. The function is further in the class $\mathsf{BQ_UL}$ if there is no intermediate measurements, i.e. the range of $\ud$ is $\mathcal{G}_S$. We define $\mathsf{promiseBQ_QL}$ and $\mathsf{promiseBQ_UL}$ similarly, but the domain of each $f_n$ can be a subset of $\{0,1\}^n$.
	
	\paragraph{Quantum learning algorithm} For a quantum online learning algorithm with $\Gamma$ being the set of samples, there exists a mapping $\ul:\Gamma\rightarrow\mathcal{C}_{m,m}$ such that $\Phi_i=\ul(z_i)$ where $z_i\in\Gamma$ is the sample received in the $i$-th step, and each entry of $K(\ul(z_i))$ can be computed deterministically in time $O(T)$ and space $O(S)$.
	
	Let $\mathcal{P}(\Gamma)$ be the collection of all probability distributions over $\Gamma$. For any distribution $D\in\mathcal{P}(\Gamma)$, let $D^T$ be $T$ i.i.d copies of $D$, so that $z\sim D^T$ means that each sample $z_i$ is independently drawn from $D$. Let $\mathcal{X},\mathcal{Y}$ be two disjoint subsets of $\mathcal{P}(\Gamma)$. The quantum learning algorithm distinguishes $\mathcal{X}$ and $\mathcal{Y}$ with error $\varepsilon$ if:
	\begin{align*}
		\forall D\in\mathcal{X},&\quad \e_{z\sim D^T}[\Tr[\rho_T\mt M_0\mt]]\geq 1-\varepsilon\\
		\forall D\in\mathcal{Y},&\quad \e_{z\sim D^T}[\Tr[\rho_T\mt M_0\mt]]\leq \varepsilon.
	\end{align*}
	And for $\varepsilon=1/3$, we simply say that the quantum learning algorithm distinguishes $\mathcal{X}$ and $\mathcal{Y}$.
	
	\paragraph{Other specifications} Notice that even in the unitary algorithms where intermediate measurements are not generally allowed, a constant number of intermediate measurements are still available because of the principle of deferred measurements (see e.g. \cite[Section 4.4]{nielsen2002quantum}), which will only increase the time and space by a constant. This means the error $\varepsilon$ in both definitions above can be safely amplified to any constant power, and the specific constant error $1/3$ can be replaced by any constant in $[0,1/2)$.
	
	The (randomized) classical counterparts of the decision and learning algorithms are defined similarly: we only need to replace quantum channels, quantum gates and measurements with stochastic matrices, classical gates and random bits respectively. For constructible functions $t(n)=\Omega(n)$ and $s(n)=\Omega(\log n)$, define $\mathsf{BPTISP}(t(n),s(n))$ as the class of boolean functions families that can be decided by a classical randomized algorithm with time $O(t(n))$ and space $O(s(n))$, and $\mathsf{promiseBPTISP}(t(n),s(n))$ accordingly. The classical randomized logspace class is defined as $\mathsf{(promise)BPL}=\mathsf{(promise)BPTISP}(\poly(n),\log(n))$.
	
\subsection{Quantum Algorithms in Prior Works}
	
	\paragraph{Hamiltonian simulation}
	Given the Hamiltonian $H$, which is a Hermitian matrix that describes the evolution of the quantum system, the \textit{Hamiltonian simulation} problem asks to simulate the evolution $e^{iHt}$ for arbitrary $t>0$. Ta-Shma showed in \cite{ta2013inverting} how to perform Hamiltonian simulation with a space-efficient unitary quantum circuit, and it is in-place (without any ancillas). We restate the result for Hermitian contractions:
	\begin{theorem}[\cite{ta2013inverting}]\label{thm:hamil}
		Given a Hermitian contraction $H\in\cc^{m\times m}$ and $\varepsilon>0$, there is a unitary quantum circuit $U$ with time $\poly(m/\varepsilon)$ and space $O(\log(m/\varepsilon))$, such that $\|U-e^{iH}\|\leq\varepsilon$.
	\end{theorem}
	We note that the dependence on $\varepsilon$ is improved in more recent Hamiltonian simulation algorithms \cite{berry2015simulating,fefferman2018complete}, however it is not required here.
	
	\paragraph{Phase estimation} Given the dimension $m$ and the error parameter $\varepsilon>0$, the \textit{phase estimation} circuit (see e.g. \cite[Section 5.2]{nielsen2002quantum}) acts on an input register of dimension $m$ and an estimation register of dimension $2^\ell=O(1/\varepsilon)$. The circuit is with time $O(2^\ell)$ and space $O(\ell+\log m)$, and accesses $2^\ell$ oracle calls to the controlled-$U$ gates, where $U\in\un_m$ is an arbitrary unitary matrix. For each $j\in[2^\ell]$, define $\lambda(j)=2j\pi/2^\ell-\pi$, and for any $\lambda\in[-\pi,\pi]$, let $J(\lambda)=\{j\in[2^\ell]\mid |\lambda(j)-\lambda|\leq\varepsilon\}$. If $v$ is a unit eigenvector of $U$ with eigenvalue $e^{i\lambda}$, the circuit maps $v\otimes\ket{0^\ell}$ to
	\[\sum_{j=0}^{2^\ell-1}\alpha_jv\otimes\ket{j},\]
	so that
	\[\sum_{j\in J(\lambda)} |\alpha_j|^2\geq 1-\varepsilon^2.\]
	\begin{remark}
		Notice that as long as $2^\ell>\pi\varepsilon^{-1}$, $|J(\lambda)|>0$. By properly choosing $\varepsilon$ and $\ell$, we can also have $|J(\lambda)|=1$ for almost every $\lambda$, but in the boundary cases we might still have $|J(\lambda)|=2$. The \textit{consistent phase estimation} in \cite{ta2013inverting} ensures that $|J(\lambda)|=1$ for every $\lambda$, but the process introduces randomness which is not allowed in unitary quantum algorithms, and the uniqueness of $j$ in $J(\lambda)$ is indeed not required for our analysis.
	\end{remark}
	
	Given a Hermitian contraction $H\in\cc^{m\times m}$, let $P_H$ be the above phase estimation circuit with $U=e^{iH}$, and $P_{H,\varepsilon}$ be the above phase estimation circuit where $U$ is replaced with the unitary quantum circuit constructed in \cref{thm:hamil} with error $2^{-\ell}\varepsilon$. Notice that $P_{H,\varepsilon}$ is a unitary quantum circuit with $\poly(m/\varepsilon)$ and space $O(\log(m/\varepsilon))$, and by \cref{lemma:error} we have $\|P_{H,\varepsilon}-P_H\|\leq \varepsilon$.
	
	Since $H$ only has eigenvalues in $[-1,1]$, we slightly modify the definition of $\lambda(j)$ so that it's truncated at $\pm1$, that is
	\[\lambda(j)=\left\{\begin{array}{ll}
		2j\pi/2^\ell-\pi & \textrm{ if }2j\pi/2^\ell-\pi\in[-1,1] \\
		\mathrm{sgn}(2j\pi/2^\ell-\pi) & \textrm{ otherwise }
	\end{array}\right.\]
	which will only make $J(\lambda)$ larger for $\lambda\in[-1,1]$.
	
	Every unit eigenvector $v$ of $H$ with eigenvalue $\lambda$ is also a unit eigenvector of $e^{iH}$ with eigenvalue $e^{i\lambda}$. Therefore for any two unit eigenvectors $u,v$ of $H$, we have
	\[\big(u\ct\otimes\bra{j}\big)P_H\mt\big(v\otimes\ket{0^\ell}\big)=\left\{\begin{array}{ll}
	\alpha_j & \textrm{ if }u=v \\
	0 & \textrm{ if }u\perp v.
	\end{array}\right.\]
	In other words, since $P_H$ is unitary,
	\[\big(u\ct\otimes\bra{0^\ell}\big)P_H^{-1}\big(v\otimes\ket{j}\big)=\left\{\begin{array}{ll}
	\overline{\alpha_j} & \textrm{ if }u=v \\
	0 & \textrm{ if }u\perp v.
	\end{array}\right.\]
	That means the projection of $P_H^{-1}\big(v\otimes\ket{j}\big)$ onto $\cc^m\otimes\ket{0^\ell}$ is along $v\otimes\ket{0^\ell}$ and has amplitude $\overline{\alpha_j}$. Combing the above observations we get the following lemma:
	
	\begin{lemma}\label{col:phase}
		Given a Hermitian contraction $H\in\cc^{m\times m}$ and $\varepsilon>0$, there is a unitary quantum circuit $P_{H,\varepsilon}$ with time $\poly(m/\varepsilon)$ and space $O(\log(m/\varepsilon))$ that is $\varepsilon$-close to a unitary operator $P_H$, which satisfies the following: There is a parameter $\ell=O(\log(1/\varepsilon))$, such that if $v$ is a unit eigenvector of $H$ with eigenvalue $\lambda\in[-1,1]$, then
		\[P_H\mt(v\otimes\ket{0^\ell})=\sum_{j=0}^{2^\ell-1}\alpha_jv\otimes\ket{j},
		\textrm{ where }\sum_{j\in J(\lambda)} |\alpha_j|^2\geq 1-\varepsilon^2.\]
		Moreover, for every $j\in[2^\ell]$,
		\[P_H^{-1}(v\otimes\ket{j})=\overline{\alpha_j}v\otimes\ket{0^\ell}+\ket{\bot},\]
		where $\ket{\bot}$ is a vector orthogonal to $\cc^m\otimes\ket{0^\ell}$.
	\end{lemma}

	\paragraph{Pure State Preparation} Our results involve the simplest form of the \textit{quantum state preparation} problem, which is to map the initial state $\ket{0^S}$ to a given pure state. Combined with the efficient Solovay-Kitaev Theorem in \cite{melkebeek2012time} we have the following:
	\begin{lemma}\label{thm:prep}
		Given $m=2^S$, a unit vector $v\in\cc^m$ and $\varepsilon>0$, there is unitary quantum circuit $Q_v$ on $S$ qubits with time $O(m\cdot\polylog(1/\varepsilon))$ and space $O(\log(m/\varepsilon))$ such that $\|Q_v\ket{0^S}-v\|_2\leq\varepsilon$.
	\end{lemma}
	\begin{proof}
		The circuit is a composition of $m-1$ two-level unitaries, that is, unitaries that operate on two dimensions of the computational basis). More specifically, starting from the initial state $\ket{0^S}$, in the $i$-th step the unitary $U_{a_i}=\begin{pmatrix}
		a_i & \sqrt{1-|a_i|^2} \\
		\sqrt{1-|a_i|^2} & \overline{a_i}
		\end{pmatrix}$ is applied on the $i$-th and $(i+1)$-th dimension, where $a_i=v_i/\sqrt{|v_i|^2+|v_{i+1}|^2+\cdots+|v_m|^2}$, so that the state after the $i$-th step is
		\[(v_1,\ldots,v_i,\sqrt{|v_{i+1}|^2+\cdots+|v_m|^2},0,\ldots,0).\]
		By \cite{ta2013inverting}, each two-level unitary can be implemented up to error $\varepsilon/m$ with time $O(m\cdot\polylog(1/\varepsilon))$ and space $O(\log(m/\varepsilon))$. By \cref{lemma:error} the total error is at most $\varepsilon$.
	\end{proof}
	For the application of the lemma in \cref{thm:unitals}, we note that the classical computation of $a_i$'s can be streamlined and incorporated into the quantum circuit, as one can maintain the tail sum $|v_i|^2+|v_{i+1}|^2+\cdots+|v_m|^2$ by starting from $1$ and subtract $|v_i|^2$ in each step. We also note that the result in \cref{thm:prep} was known in \cite{mottonen2005transformation} with a slightly different proof, and while the space complexity is not originally stated there, it is implicit from their construction.

%% file: implement.tex
\section{Quantum Implementations of Contractions}\label{sect:impl}

\begin{lemma}\label{thm:implem}
	Given a contraction $A\in\cc^{m\times m}$ and $\varepsilon>0$, there is a unitary quantum circuit $Q_A$ with time $\poly(m/\varepsilon)$ and space $O(\log(m/\varepsilon))$, and a parameter $\ell=O(\log(1/\varepsilon))$, such that for unit vector $v$ of dimension $4m$, $\|Q_A\mt(v\otimes\ket{0^\ell})-(V_A\mt v)\otimes\ket{0^\ell}\|_2\leq\varepsilon$, where
	\[V_A=\diag(U_A,U_{A\ct})=\begin{pmatrix}
		A & \sqrt{\id_m-AA\ct} & &\\
		\sqrt{\id_m-A\ct A} & -A\ct & &\\
		 & & A\ct & \sqrt{\id_m-A\ct A} \\
		 & & \sqrt{\id_m-AA\ct} & -A
	\end{pmatrix}\]
\end{lemma}

\begin{proof}
	Let $H$ be the Hermitian contraction $\begin{pmatrix} & A \\ A\ct & \end{pmatrix}$. Notice that
	\[U_H=\begin{pmatrix}
		H & \sqrt{\id_{2m}-H^2} \\
		\sqrt{\id_{2m}-H^2} & -H
	\end{pmatrix}=\begin{pmatrix}
		 & A & \sqrt{\id_m-AA\ct} &\\
		A\ct & & & \sqrt{\id_m-A\ct A}\\
		\sqrt{\id_m-AA\ct} & & & -A\\
		 & \sqrt{\id_m-A\ct A} & -A\ct & 
	\end{pmatrix}\]
	which differs from $V_A$ only by permutations:
	\[V_A=\begin{pmatrix}
		\id_m & & & \\
		& & & \id_m \\
		& \id_m & & \\
		& & \id_m &
	\end{pmatrix}\cdot U_H\cdot\begin{pmatrix}
	& & \id_m & \\
	\id_m & & & \\
	& \id_m & &\\
	& & & \id_m 
	\end{pmatrix}\]
	Since the permutations are only on two qubits, it suffices to implement $U_H$ on $v$ up to error $\varepsilon$. 
	
	Let $v=\begin{pmatrix}
		v_1 \\ v_2
	\end{pmatrix}$ where both $v_1$ and $v_2$ are of dimension $2m$. Suppose $H$ has the eigen decomposition $H=\sum\limits_{k=1}^{2m}\lambda_k\mt u_k\ct u_k\mt$, and $v_1,v_2$ are decomposed into this eigenbasis as 
	\[v_1=\sum_{k=1}^{2m}\omega_k^{(0)}u_k\mt,\quad v_2=\sum_{k=1}^{2m}\omega_k^{(1)}u_k\mt,\quad\textrm{where }
	\sum_{k=1}^{2m}\left|\omega_k^{(0)}\right|^2+\sum_{k=1}^{2m}\left|\omega_k^{(1)}\right|^2=1.\]
	Since $v$ can be written as $\ket{0}\otimes v_1+\ket{1}\otimes v_2$, applying the phase estimation circuit $P_{H,\varepsilon_1}$ in \cref{col:phase} on $v\otimes\ket{0^\ell}$ results in:
	\begin{align*}
		&\sum_{k=1}^{2m}\sum_{j=0}^{2^\ell-1} \omega_k^{(0)}\alpha_{j,k}\mt\ket{0}\otimes u_k\mt\otimes\ket{j}+
		\sum_{k=1}^{2m}\sum_{j=0}^{2^\ell-1} \omega_k^{(1)}\alpha_{j,k}\mt\ket{1}\otimes u_k\mt\otimes\ket{j}+ \ket{\epsilon_1} \\
		=& \sum_{k=1}^{2m}\sum_{j\in J(\lambda_k)} \omega_k^{(0)}\alpha_{j,k}\mt\ket{0}\otimes u_k\mt\otimes\ket{j}+
		\sum_{k=1}^{2m}\sum_{j\in J(\lambda_k)} \omega_k^{(1)}\alpha_{j,k}\mt\ket{1}\otimes u_k\mt\otimes\ket{j} + \ket{\epsilon_2}.
	\end{align*}
	where for each $k$ it holds $\sum_{j\in J(\lambda_k)}|\alpha_{j,k}|^2\geq 1-\varepsilon_1^2$. Here $\varepsilon_1$ is an error parameter to be determined later, and $\ell=O(\log(1/\varepsilon_1))$. The error vector $\ket{\epsilon_1}$ is introduced due to the difference between $P_{H,\varepsilon_1}$ and $P_H$, and thus $\|\ket{\epsilon_1}\|_2\leq \|P_{H,\varepsilon_1}-P_H\|\leq\varepsilon_1$. The error vector $\ket{\epsilon_2}-\ket{\epsilon_1}$ is a weighted sum of $4m$ orthogonal error vectors, with lengths at most $\varepsilon_1$ and weights $\omega_k^{(0)},\omega_k^{(1)}$, and thus has length at most $\varepsilon_1$. Therefore  $\|\ket{\epsilon_2}\|_2\leq 2\varepsilon_1$.
	
	Now apply the following unitary transformation on the first qubit and last $\ell$ qubits:
	\begin{align*}
		\ket{0}\ket{j}& \rightarrow \lambda(j)\ket{0}\ket{j}+\sqrt{1-\lambda(j)^2}\ket{1}\ket{j} \\
		\ket{1}\ket{j}& \rightarrow \sqrt{1-\lambda(j)^2}\ket{0}\ket{j}-\lambda(j)\ket{1}\ket{j}
	\end{align*}
	which gives
	\begin{align*}
		& \sum_{k=1}^{2m}\sum_{j\in J(\lambda_k)} \omega_k^{(0)}\alpha_{j,k}\mt 
		\left[\lambda(j)\ket{0}+\sqrt{1-\lambda(j)^2}\ket{1}\right]\otimes u_k\mt\otimes\ket{j} \\
		+& \sum_{k=1}^{2m}\sum_{j\in J(\lambda_k)} \omega_k^{(1)}\alpha_{j,k}\mt
		\left[\sqrt{1-\lambda(j)^2}\ket{0}-\lambda(j)\ket{1}\right]\otimes u_k\mt\otimes\ket{j}+\ket{\epsilon_3}
	\end{align*}
	This unitary transformation can be implemented as a serial combination of $2^\ell$ single-qubit unitaries $U_{\lambda(j)}$ controlled by the last $\ell$ qubits representing $j$. Each one of them can be constructed up to error $2^{-\ell}\varepsilon_1$ in time $\polylog(1/\varepsilon_1)$ and space $O(\log(1/\varepsilon_1))$ by \cite[Theorem 7]{melkebeek2012time}. Therefore by \cref{lemma:error} we have $\|\ket{\epsilon_3}\|_2\leq \|\ket{\epsilon_2}\|_2+\varepsilon_1\leq 3\varepsilon_1$.
	
	Finally applying the reverse phase estimation $P_{H,\varepsilon_1}^{-1}$ gives the following state, where $\ket{\bot}$ is orthogonal to $\cc^2\otimes\cc^{2m}\otimes\ket{0^\ell}$:
	\begin{align*}
		& \sum_{k=1}^{2m}\sum_{j\in J(\lambda_k)} |\alpha_{j,k}|^2\omega_k^{(0)}
		\left[\lambda(j)\ket{0}+\sqrt{1-\lambda(j)^2}\ket{1}\right]
		\otimes u_k\mt\otimes\ket{0^\ell} \\
		&+ \sum_{k=1}^{2m}\sum_{j\in J(\lambda_k)} |\alpha_{j,k}|^2\omega_k^{(1)}
		\left[\sqrt{1-\lambda(j)^2}\ket{0}-\lambda(j)\ket{1}\right]
		\otimes u_k\mt\otimes\ket{0^\ell}
		+\ket{\epsilon_4}+\ket{\bot} \\
		= & \sum_{k=1}^{2m}\sum_{j\in J(\lambda_k)} |\alpha_{j,k}|^2 \omega_k^{(0)}
		\left[\lambda_k\ket{0}+\sqrt{1-\lambda_k^2}\ket{1}\right]
		\otimes u_k\mt\otimes\ket{0^\ell} \\
		&+ \sum_{k=1}^{2m}\sum_{j\in J(\lambda_k)} |\alpha_{j,k}|^2 \omega_k^{(1)}
		\left[\sqrt{1-\lambda_k^2}\ket{0}-\lambda_k\ket{1}\right]
		\otimes u_k\mt\otimes\ket{0^\ell}+\ket{\epsilon_5}+\ket{\bot} \\
		= & \sum_{k=1}^{2m}\omega_k^{(0)}
		\left[\lambda_k\ket{0}+\sqrt{1-\lambda_k^2}\ket{1}\right]\otimes u_k\mt\otimes\ket{0^\ell} \\
		&+ \sum_{k=1}^{2m}\omega_k^{(1)}
		\left[\sqrt{1-\lambda_k^2}\ket{0}-\lambda_k\ket{1}\right]\otimes u_k\mt\otimes\ket{0^\ell}+\ket{\epsilon_6}+\ket{\bot} \\
		= & \sum_{k=1}^{2m}\omega_k^{(0)}\big[U_H\mt(\ket{0}\otimes u_k\mt)\big]\otimes\ket{0^\ell}
		+ \sum_{k=1}^{2m}\omega_k^{(1)}\big[U_H\mt(\ket{1}\otimes u_k\mt)\big]\otimes\ket{0^\ell}+\ket{\epsilon_6}+\ket{\bot} \\
		= & \big[U_H\mt(\ket{0}\otimes v_1)\big]\otimes\ket{0^\ell}
		+\big[U_H\mt(\ket{1}\otimes v_2)\big]\otimes\ket{0^\ell}+\ket{\epsilon_6}+\ket{\bot} \\
		= & (U_H\mt v)\otimes\ket{0^\ell}+\ket{\epsilon_6}+\ket{\bot}.
	\end{align*}
	Here $\|\ket{\epsilon_4}\|_2\leq \|\ket{\epsilon_3}\|_2+\|P_{H,\varepsilon_1}^{-1}-P_H^{-1}\|\leq 4\varepsilon_1$.
	Also, similar to the reasoning for $\ket{\epsilon_2}-\ket{\epsilon_1}$, since for every $k$, $1-\varepsilon_1^2\leq \sum_{j\in J(\lambda_k)} |\alpha_{j,k}|^2\leq 1$, and for every $j\in J(\lambda_k)$, $\|U_{\lambda(j)}-U_{\lambda_k}\|_2\leq |\lambda(j)-\lambda_k|\leq \varepsilon_1$, we have 
	\[\|\ket{\epsilon_6}\|_2\leq \|\ket{\epsilon_5}\|_2+\varepsilon_1^2
	\leq \|\ket{\epsilon_4}\|_2+\varepsilon_1+\varepsilon_1^2\leq 6\varepsilon_1.\]
	Finally, notice that both $(U_H\mt v)\otimes\ket{0^\ell}$ and $(U_H\mt v)\otimes\ket{0^\ell}+\ket{\epsilon_6}+\ket{\bot}$ are unit vectors, while $\ket{\bot}$ is orthogonal to $(U_H\mt v)\otimes\ket{0^\ell}$, so we have
	\[\big|\big((U_H\mt v)\ct\otimes\bra{0^\ell}\big)\big((U_H\mt v)\otimes\ket{0^\ell}+\ket{\epsilon_6}+\ket{\bot}\big)\big|=\big|1+\big((U_H\mt v)\ct\otimes\bra{0^\ell}\big)\ket{\epsilon_6}\big|\geq 1-\|\ket{\epsilon_6}\|_2,\]
	which implies that $\|\ket{\epsilon_6}+\ket{\bot}\|_2\leq\sqrt{2\|\ket{\epsilon_6}\|_2}$. Therefore it suffices to take $\varepsilon_1\leq \varepsilon^2/12$, and the theorem follows.
\end{proof}

As a by product, when we take the contraction $A$ in \cref{thm:implem} to be unitary, we get the unitary implementation of any unitary matrix, with the number of ancillas only depending on the error:
\begin{corollary}\label{col:implem}
	Given a unitary matrix $U\in\un_m$ and $\varepsilon>0$, there is a unitary quantum circuit $Q_U$ with time $\poly(m/\varepsilon)$ and space $O(\log(m/\varepsilon))$, and a parameter $\ell=O(\log(1/\varepsilon))$, such that for any unit vector $v$ of dimension $m$, $\|Q_U\mt(v\otimes\ket{0^\ell})-(Uv)\otimes\ket{0^\ell}\|_2\leq\varepsilon$.
\end{corollary}
\begin{proof}
	Use the exact same circuit in \cref{thm:implem} by adding two ancilla qubits to $v$ initialized at $\ket{00}$. Notice that $V_U=\diag(U,-U\ct,U\ct,-U)$, and thus the output state is $\varepsilon$ close to $[V_U(\ket{00}\otimes v)]\otimes\ket{0^\ell}=\ket{00}\otimes (Uv)\otimes\ket{0^\ell}$. Rearranging the order of qubits and the claim follows.
\end{proof}

Finally, for permutation matrices, we present a simple unitary implementation without any ancillas by decomposing it into transpositions. The following theorem will be used in the proof of \cref{thm:power} but not necessarily, as \cref{col:implem} is enough for the job. Nevertheless, the decomposition of permutations in logarithmic space might be of independent interest, so we defer the proof to Appendix~\ref{app:perm}.

\begin{lemma}\label{thm:perm}
	Given a permutation $\sigma\in S_m$ and $\varepsilon>0$, there is a unitary quantum circuit $U$ with time $\poly(m/\varepsilon)$ and space $O(\log(m/\varepsilon))$, such that $\|U-P_\sigma\|\leq\varepsilon$, where $P_\sigma\in\{0,1\}^{m\times m}$ is the matrix representation of $\sigma$.
\end{lemma}

%% file: unital.tex
\section{Contraction Powering in Quantum Logspace}\label{sect:cpowering}

\begin{definition}[Contraction Powering]
	Given $m=2^S$, a contraction $A\in\cc^{m\times m}$, a positive integer $T$ in unary, and two vectors $v,w\in\cc^m$ with $\|v\|_2=\|w\|_2=1$ as the input, it is promised that $|w\ct A^Tv|^2$ is either in $[0,1/3]$ or $[2/3,1]$, and the goal of the \textsc{ContractionPowering} problem is to distinguish between the two cases.
\end{definition}

\begin{theorem}\label{thm:power}
	$\textsc{ContractionPowering}\in\mathsf{promiseBQ_UL}$. Moreover, given the same input \\ $(m,A,T,v,w)$ but without the promise on $|w\ct A^Tv|^2$, while also given an error parameter $\varepsilon>0$, there is a unitary quantum circuit $W$ with time $\poly(mT/\varepsilon)$ and space $S'=O(\log(mT/\varepsilon))$ such that $|\bra{0^{S'}}W\ket{0^{S'}}|^2$ is $\varepsilon$-close to $|w\ct A^Tv|^2$.
\end{theorem}
\begin{proof}
	First, let $Q_v$ and $Q_w$ be the circuits preparing states $v$ and $w$ with error $\varepsilon/8$ in \cref{thm:prep} respectively. Since 
	\[\Big||\bra{0^S}Q_w\ct A^TQ_v\mt\ket{0^S}|^2-|w\ct A^Tv|^2\Big|\leq4\|Q_v\ket{0^S}-v\|_2+4\|Q_w\ket{0^S}-w\|_2\leq\varepsilon/2,\]
	in the rest of the proof we can safely assume that $Q_v\ket{0^S}=v$ and $Q_w\ket{0^S}=w$ while halving $\varepsilon$.

	Intuitively, our algorithm iteratively applies the unitary quantum circuit in \cref{thm:implem} for $T$ times. However, since \cref{thm:implem} only implements $V_A$ instead of $A$, we have to `throw away' the unwanted dimensions introduced by $V_A$, by permuting them into additional dimensions.
	
	Formally, let $\ell=O(\log (T/\varepsilon))$ be the one in \cref{thm:implem} with error parameter $(2T)^{-1}\varepsilon$. The circuit works on three parts of qubits: the counter register $C$ of dimension $2T$, the vector register of dimension $m$, and $\ell$ ancilla qubits. The circuit starts by preparing $\ket{0}_C\mt\otimes v\otimes\ket{0^\ell}$ by applying $Q_v$. Then repeat the following two steps for $T$ times:
	\begin{enumerate}
		\item Apply $V_A$ on the last two qubits of the timer register and the entire vector register by \cref{thm:implem};
		\item Apply the permutation 
		\begin{align*}
			& \ket{0}\rightarrow\ket{0},\quad\ket{2T-2}\rightarrow\ket{1},\quad\ket{2T-1}\rightarrow\ket{2}\\
			& \ket{i}\rightarrow\ket{i+2}\ \forall i=1,\ldots,2T-3.
		\end{align*}
		on the counter register by \cref{thm:perm}.
	\end{enumerate}
	Finally, apply $Q_w\ct$ on the vector register and measure with the projection onto $\ket{0}_C\mt\otimes\ket{0^S}\otimes\ket{0^\ell}$.
	
	To prove the correctness of the algorithm, we first assume that all the implementations in \cref{thm:implem} and \cref{thm:perm} are errorless, i.e. the evolution is completely within the subspace $\cc^{2T}\otimes\cc^m\otimes\ket{0^\ell}$. Then it suffices to notice that $V_A$ is block-diagonal, so that step 1 acts locally on the $T$ subspaces spanned by $\ket{2i}_C\mt$ and $\ket{2i+1}_C\mt$. Therefore after the $i$-th application of $V_A$, the projection of the current state onto $\ket{j}_C\mt$ is always $0$ for $j\geq 2i$, and thus before each application of $V_A$, the projection onto $\ket{1}_C\mt$ is always $0$. So the state after the $i$-th repetition is $\ket{0}_C\mt\otimes (A^iv)+\ket{\bot}$, where $\ket{\bot}$ is orthogonal to $\ket{0}_C\mt$. The output probability is then
	\[\left|\left(\bra{0}_C\mt\otimes\bra{0^S}\right)\left(\id_{2T}\mt\otimes U_w\ct\right)\left(\ket{0}_C\mt\otimes (A^Tv)+\ket{\bot}\right)\right|^2=|w\ct A^Tv|^2.\]
	
	Now each step in the repetition introduces an error of $(2T)^{-1}\varepsilon$. Therefore, by \cref{lemma:error}, the total error of the unitary quantum circuit $W$, compared to the above ideal case, is at most $\varepsilon$.
\end{proof}

\section{Equivalence of Unital and Unitary Quantum Logspace}\label{sect:unital}

\subsection{Simulating Unital Quantum Logspace with Constant Error}

\begin{lemma}\label{thm:unital}
	Given a unital quantum algorithm with time $T$ and space $S=\log m$ specified by the natural representations $K(\Phi_1),\ldots,K(\Phi_T\mt)\in\cc^{m^2\times m^2}$, and an error parameter $\varepsilon>0$, there is a unitary quantum circuit with time $\poly(mT/\varepsilon)$ and space $O(\log(mT/\varepsilon))$, such that if the original unital circuit outputs $0$ with probability $p$, then the unitary circuit outputs $0$ with probability $\sin^2(p+\alpha)$, where $|\alpha|\leq \varepsilon$.
\end{lemma}
\begin{proof}
	We can always assume that $S$ is an odd number and $m\geq \max(4\varepsilon^{-1},8)$ by adding dummy dimensions. As each $K(\Phi_i)$ is a contraction, the following matrix $A$ of dimension $m^2T$ is also a contraction:
	\[A=\begin{pmatrix}
		 & & & & K(\Phi_T\mt) \\
		K(\Phi_1) & & & & \\
		& K(\Phi_2) & & & \\
		& & \ddots & & \\
		& & & K(\Phi_{T-1}\mt) &
	\end{pmatrix}\]
	Since the final state of the unital quantum algorithm is $\rho_T\mt=\Phi_T\mt\circ\cdots\circ \Phi_1\mt(\rho_0\mt)$, we can rewrite the output probability of the unital quantum algorithm as
	\begin{align*}
		p=\Tr[\rho_T\mt M_0] & = \vect(M_0)\ct\vect(\rho_T\mt)=\vect(M_0)\ct K(\Phi_T\mt)\cdots K(\Phi_1\mt)\vect(\rho_0\mt) \\
		& = \left(\vect(M_0)\ct\otimes\bra{0}\right) A^T \big(\vect(\rho_0\mt)\otimes\ket{0}\big)
	\end{align*}

	Let $v=\vect(\rho_0\mt)\otimes\ket{0}$ which is already a unit vector. Since $\|\vect(M_0)\|_2=\sqrt{m/2}$, let $w=\sqrt{\frac{2}{m}}\vect(M_0)\otimes\ket{0}$. Let $\varepsilon_1$ be the error parameter to be determined later. \cref{thm:power} constructs a unitary quantum algorithm $W$ with time $\poly(mT/\varepsilon_1)$ and space $S'=O(\log(mT/\varepsilon_1))$, such that $|\bra{0^{S'}}W\ket{0^{S'}}|^2$ is $\varepsilon_1$-close to $2p^2/m$. Therefore $|\bra{0^{S'}}W\ket{0^{S'}}|$ is $\sqrt{\varepsilon_1}$-close to $\sqrt{\frac{2}{m}}p$.

	Let
	\[R=\left(\mathbf{I}_{S'}-2W\ket{0^{S'}}\bra{0^{S'}}W\ct\right)
	\left(\mathbf{I}_{S'}-2\ket{0^{S'}}\bra{0^{S'}}\right)\]
	be the rotation on the subspace spanned by $\ket{0^{S'}}$ and $W\ket{0^{S'}}$, of degree $2\cos^{-1}|\bra{0^{S'}}W\ket{0^{S'}}|$. By the estimation
	\[\frac{\pi}{2}-x-\frac{1}{4}x^3\leq\cos^{-1}x\leq\frac{\pi}{2}-x,\forall x\in[0,1],\]
	and since $(x+y)^3\leq4(x^3+y^3)$ for non-negative $x,y$, it can be calculated that the degree of the rotation $R$ is in the range
	\[\left[\pi-2\sqrt{\frac{2}{m}}p-4\sqrt{\varepsilon_1}-\frac{4}{m}\sqrt{\frac{2}{m}},\ \pi-2\sqrt{\frac{2}{m}}p+2\sqrt{\varepsilon_1}\right].\]
	
	Since $S$ is an odd number, $k=\sqrt{m/8}$ is an integer. Applying $R$ for $k$ times will rotate $\ket{0^{S'}}$ by a degree of $k\pi-p-\alpha$, where $|\alpha|\leq\sqrt{2m\varepsilon_1}+2m^{-1}$. Therefore the projective measurement of the state $R^k\ket{0^{S'}}$ onto the subspace orthogonal to $\ket{0^{S'}}$ outputs $0$ with probability $\sin^2(p+\alpha)$. Let $\varepsilon_1=(8m)^{-1}\varepsilon^2$, and notice that $2m^{-1}\leq\varepsilon/2$, so that we have $|\alpha|<\varepsilon$, and the circuit $R^k$ is unitary with time $\poly(mT/\varepsilon)$ and space $O(\log(mT/\varepsilon))$.
\end{proof}

\begin{theorem}\label{thm:main}
	$\mathsf{BQ_QL}=\mathsf{BQ_UL}$, and $\mathsf{promiseBQ_QL}=\mathsf{promiseBQ_UL}$.
\end{theorem}
\begin{proof}
	Clearly $\mathsf{BQ_QL}\supseteq\mathsf{BQ_UL}$, and $\mathsf{promiseBQ_QL}\supseteq\mathsf{promiseBQ_UL}$. To prove the other direction, notice that quantum circuits are unital, therefore by \cref{thm:unital} with $\varepsilon=0.01$ they can be simulated by unitary quantum circuits with polynomial time and logarithmic space. Since the original output probability $p$ is promised to be in $[0,1/3]$ or $[2/3,1]$, the value of $\sin^2(p+\alpha)$ is in $[0,0.12]$ or $[0.37,1]$ respectively, and thus it suffices to perform a constant rounds of amplification in order to bring the error down to less than $1/3$.
\end{proof}

\begin{remark}
	Though we proved \cref{thm:main} via the contraction powering algorithm, the unitary quantum circuit that simulates a given quantum circuit with intermediate measurements can be more simply constructed without using \cref{thm:implem}. In details, given a channel $\Phi$ in the quantum circuit, we can directly write out the natural representations $K(\Phi)$, and apply the matrix on the vectorized density matrix $\vect(\rho)$:
	\begin{itemize}
		\item[-] If $\Phi$ is a unitary quantum gate $U$, then $K(\Phi)=\overline{U}\otimes U$ which can be implemented by applying $U$ and then $\overline{U}$;
		\item[-] If $\Phi$ is a single-qubit measurement, then $K(\Phi)$ is a diagonal matrix with diagonal entries in $\{0,1\}$. It can be implemented using a similar "permute and throw away" technique as in \cref{thm:power}, which after applied $T$ times increases the dimension (instead of the space!) by a factor of $T$.
	\end{itemize}
	And the resulting circuit can be amplified in the same way as in \cref{thm:unital}.
\end{remark}

\subsection{Simulating Unital Quantum Logspace with Small Error}
Now we can improve the result in \cref{thm:unital} to arbitrarily small error (namely the probability of outputting $0$ is $(p+\alpha)$ instead of $\sin^2(p+\alpha)$). Interestingly, the improvement relies on a stronger version of \cref{thm:power}, which in turn relies on \cref{thm:main}. In a way, we use these results to improve themselves!

We start with the stronger version of \cref{thm:power}, which outputs the numerical value of $|w\ct A^Tv|^2$ instead of outputting $0$ with such probability. Here the quantum circuit outputs a number by a final measurement over the computational basis.

\begin{lemma}\label{lemma:power}
	Given $m=2^S$, a contraction $A\in\cc^{m\times m}$, a positive integer $T$, two unit vectors $v,w\in\cc^m$ and an error parameter $\varepsilon>0$, there is a unitary quantum circuit with time $\poly(mT/\varepsilon)$ and space $O(\log(mT/\varepsilon))$ such that with probability $1-2^{-\poly(mT/\varepsilon)}$, it outputs $|w\ct A^Tv|^2$ with additive error $\varepsilon$. 
\end{lemma}
\begin{proof}[Proof]
	\cref{thm:power} provides a unitary quantum circuit $W$ with time $\poly(mT/\varepsilon)$ and space $O(\log(mT/\varepsilon))$ which outputs $0$ with probability $p$ such that $\big|p-|w\ct A^Tv|^2\big|\leq\varepsilon/2$. By Marriott-Watrous amplification \cite[Theorem 3.3]{marriott2005quantum}, there is a quantum circuit $W'$ with time $\poly(mT/\varepsilon)$ and space $O(\log(mT/\varepsilon))$ \textit{with} intermediate measurements, that uses $W$ and $W^{-1}$ as sub-circuits, and with probability $1-\delta=1-2^{-\poly(mT/\varepsilon)}$ outputs a value $\tilde{p}$ such that $|\tilde{p}-p|\leq \varepsilon/4$.
	
	Since the resulting circuit $W'$ is not unitary, we would like to use \cref{thm:main} to compute unitarily each bit in the output value $\tilde{p}$ of $W'$. Furthermore, using the result in \cite{fefferman2016space} that $\mathsf{BQ_UL}=\mathsf{Q_UL}(1-2^{-\poly(n)},2^{-\poly(n)})$ (which stands for unitary quantum logspace with exponentially small error) the total error probability can be reduced down to $2^{-\poly(mT/\varepsilon)}$. Assuming that every bit in $\tilde{p}$ is $0$ with probability either in $[0,1/3]$ or $[2/3,1]$, then for $1\leq i\leq\lceil\log(1/\varepsilon)\rceil+2$, we let $W_i$ be the unitary quantum circuit that computes the $i$-th bit of $\tilde{p}$ with exponentially small error. Ideally, the outputs of $W_i$ combined together would $\varepsilon$-approximate $|w\ct A^Tv|^2$.
	
	However, the value $\tilde{p}$ outputted by the Marriott-Watrous amplification might be different in each $W_i$, so the final approximation assembled can be totally wrong (for instance, when $p=0.5$, the outputs $\tilde{p}=\overline{0.1000\ldots}$ and $\tilde{p}=\overline{0.0111\ldots}$ might be assembled to $\overline{0.1111\ldots}$). Moreover, the error reduction in \cite{fefferman2016space} may have unpredictable results, as the promises on the distributions of the bits in $\tilde{p}$ are not guaranteed (again when $p=0.5$, the most significant bit of $\tilde{p}$ is equally distributed on $0$ and $1$).
	
	Fortunately, we can solves both problems by computing from the most significant bit to the least significant bit. We maintain a value $q\in[0,1]$ which is initialized to $0$. For each $i=1$ to $\lceil\log(1/\varepsilon)\rceil+2$ do the following: Run the modified circuit $W_i$ which outputs the $i$-th bit of $(\tilde{p}-q)$ instead of $\tilde{p}$. To deal with case when $\tilde{p}-q$ is outside of $[0,2^{-i+1})$, if $\tilde{p}-q<0$ it outputs $0$, and if $\tilde{p}-q\geq 2^{-i+1}$ it outputs $1$. Let the output bit be $b_i$ and update $q$ to $q+b_i\cdot 2^{-i}$.
	
	We claim that with probability $1-2^{-\poly(mT/\varepsilon)}$, $|q-p|\leq\varepsilon/2$. First notice that, if every bit in $\tilde{p}$ is $0$ with probability in $[0,2\delta]\cup[1-2\delta,1]$, then the error reduction will work as intended, while with probability $1-O(\delta\log(1/\varepsilon))=1-2^{-\poly(mT/\varepsilon)}$ the value $\tilde{p}$ is the same in each circuit $W_i$, so that $q$ is also the same as $\tilde{p}$.
	
	Now let $i$ be the first index such that the $i$-th bit of $\tilde{p}$ is $0$ with probability in $[2\delta,1-2\delta]$. As the Marriott-Watrous amplification outputs incorrectly with probability at most $\delta$, it means that there are two valid outputs $\tilde{p}_1$ and $\tilde{p}_2$, both are $\varepsilon/4$-close to $p$, and they coincide in the first $i-1$ bits but differs at the $i$-th bit. Let $q_i$ be the value of $q$ at that step, which consists of the first $i-1$ bits of $\tilde{p}_1$ and $\tilde{p}_2$, then $|q_i+2^{-i}-p|\leq\varepsilon/4$. Therefore the remaining bits of $q$ could only be $\overline{011\ldots11}$, $\overline{100\ldots00}$ or $\overline{100\ldots01}$, which means $|q_i+2^{-i}-q|\leq\varepsilon/4$ and thus $|q-p|\leq\varepsilon/2$. Notice that on the $i$-th (and the last bit when $b_i=1$) the error reduction may fail and arbitrarily output $0$ or $1$, but it does not matter as both $0$ and $1$ are viable in these cases.
	
	As a conclusion, the value $q$ is an $\varepsilon$-approximation of $|w\ct A^Tv|^2$ with probability $1-2^{-\poly(mT/\varepsilon)}$. The above circuit that outputs $q$ is clearly with time $\poly(mT/\varepsilon)$ and space $O(\log(mT/\varepsilon))$ as we use constructions in \cref{thm:main} and \cite{fefferman2016space}. Finally, the the circuit is unitary since the $O(\log(1/\varepsilon))$ measurements that output $b_i$'s can be deferred, and each $W_i$ can be uncomputed by implementing the circuit in reverse.
\end{proof}

\begin{corollary}\label{col:power}
	Given $m=2^S$, a contraction $A\in\cc^{m\times m}$, a positive integer $T$, two unit vectors $v,w\in\cc^m$ and an error parameter $\varepsilon>0$, there is a unitary quantum circuit with time $\poly(mT/\varepsilon)$ and space $O(\log(mT/\varepsilon))$ such that with probability $1-2^{-\poly(mT/\varepsilon)}$, it outputs $w\ct A^Tv$ with additive error $\varepsilon$. 
\end{corollary}
\begin{proof}
	Let	$A_1=\begin{pmatrix}A& \\ & 1\end{pmatrix}$, 
	$v_1=\begin{pmatrix}v/\sqrt{2}\\1/\sqrt{2}\end{pmatrix}$, 
	$v_1'=\begin{pmatrix}v/\sqrt{2}\\\mathrm{i}/\sqrt{2}\end{pmatrix}$ and
	$w_1=\begin{pmatrix}w/\sqrt{2}\\1/\sqrt{2}\end{pmatrix}$. Since we have
	\[w\ct A^Tv=\frac{1}{2}\left(4|w_1\ct A_1^Tv_1|^2-|w\ct A^Tv|^2-1\right)+\frac{\mathrm{i}}{2}\left(4|w_1\ct A_1^Tv_1'|^2-|w\ct A^Tv|^2-1\right),\]
	computing $|w\ct A^Tv|^2$, $|w_1\ct A_1^Tv_1|^2$ and $|w_1\ct A_1^Tv_1'|^2$ each up to error $\varepsilon/2$ gives $w\ct A^Tv$ with error $\varepsilon$.
\end{proof}

Notice that one can instead achieve $1/\poly(mT/\varepsilon)$ error probability without using the exponential error reduction in \cite{fefferman2016space}, by simply repeating the decision circuit in $\mathsf{BQ_UL}$ for $O(\log(mT/\varepsilon))$ rounds. Nevertheless, it is enough for proving the following theorem, which states that unitary quantum circuits can simulate any unital quantum algorithm by computing its output distribution with arbitrarily small error.

\begin{theorem}\label{thm:unitals}
	Given a unital quantum algorithm with time $T$ and space $S=\log m$ specified by the natural representations $K(\Phi_1),\ldots,K(\Phi_T\mt)\in\cc^{m^2\times m^2}$, where
	$\rho_T\mt=\Phi_T\mt\circ \Phi_{T-1}\mt\circ\cdots\circ \Phi_1(\ket{0^S}\bra{0^S})$
	is its final state, a multi-outcome measurement $\{M_0,\ldots,M_{r-1}\}$ over the computational basis, and an error parameter $\varepsilon>0$, there is a unitary quantum circuit $W$ with time $\poly(mT/\varepsilon)$ and space $S'=O(\log(mT/\varepsilon))$ such that if $w\in\cc^{2^{S'}}$ is the vector representation of $W\ket{0^{S'}}$ in computational basis, for every $j\in[r]$ it holds that $\big||w_j|^2-\Tr[\rho_TM_j\mt]\big|\leq\varepsilon$.
\end{theorem}
\begin{proof}
	
	For every $j\in[r]$, let $m_j$ be the dimension of the subspace that $M_j$ projects onto. In other words, $m_j=\|\vect(M_j)\|_2^2$. As in the proof of \cref{thm:unital}, we can construct a contraction $A\in\cc^{m^2T\times m^2T}$ and unit vectors $v,w\in\cc^{m^2T}$ such that $w\ct A^Tv=\Tr[\rho_TM_j\mt]/\sqrt{m_j}$. By \cref{col:power}, for every $j\in[r]$ there is a unitary quantum circuit $Q_j$ with time $\poly(mT/\varepsilon)$ and space $O(\log(mT/\varepsilon))$ such that with probability $1-2^{-\poly(mT/\varepsilon)}$, it gives an $(2m)^{-3}\varepsilon^2$-approximation of $\Tr[\rho_TM_j\mt]/\sqrt{m_j}$, which implies an $(2m)^{-1}\varepsilon$-approximation of $\sqrt{\Tr[\rho_TM_j\mt]}$.
	
	Consider the preparation circuit constructed in \cref{thm:prep} which prepares the unit vector \[u=\left(\sqrt{\Tr[\rho_TM_0\mt]},\sqrt{\Tr[\rho_TM_1\mt]},\ldots,\sqrt{\Tr[\rho_TM_{r-1}\mt]}\right).\]
	with error $\varepsilon/2$. By construction, the preparation circuit can be viewed as a composition of $r-1$ unitary operators, each controlled by a different entry in $u$. Since $u$ is not explicitly given, we instead control these unitary operators with the output qubits of $Q_j$, but without measurements. Each circuit $Q_j$ is applied in reverse after the control, so that the space can be reused. 
	
		\begin{figure}[t]
		\centering
		\begin{tikzpicture}
		
		\draw[ultra thick] (0,-0.5) -- (2,-0.5);
		\draw[ultra thick] (0,-0.25) node[left]{$u_j$} -- (2,-0.25);
		\draw[ultra thick] (0,0) -- (2,0);
		\draw[ultra thick] (0,1) -- (2,1);
		\draw[ultra thick] (1,-0.5) node[circle,fill,inner sep=2pt]{} --
		(1,-0.25) node[circle,fill,inner sep=2pt]{} -- 
		(1,0) node[circle,fill,inner sep=2pt]{} -- (1,1);
		\draw[fill=white, thick] (0.7,0.7) rectangle (1.3,1.3) node[midway] {$U$};
		\draw[dashed] (-0.5,0.4) rectangle (2.5,2.7);
		\node[label={[align=center]preparation\\circuit}] at (1,1.5) {};
		
		\node at (3.5,0.5) {\LARGE $\Rightarrow$};
		
		\begin{scope}[xshift=5cm]
		\draw[ultra thick] (0,0) -- (4,0);
		\draw[ultra thick] (0,-0.25) -- (4,-0.25);
		\draw[ultra thick] (0,-0.5) -- (4,-0.5);
		\draw[ultra thick] (0,-0.75) -- (4,-0.75);
		\draw[ultra thick] (0,-1.6) -- (4,-1.6);
		\draw[ultra thick] (0,1) -- (4,1);
		\draw[ultra thick] (2,-0.5) node[circle,fill,inner sep=2pt]{} --
		(2,-0.25) node[circle,fill,inner sep=2pt]{} --
		(2,0) node[circle,fill,inner sep=2pt]{} -- (2,1);
		\draw[fill=white, thick] (1.7,0.7) rectangle (2.3,1.3) node[midway] {$U$};
		\draw[dashed] (-0.5,0.4) rectangle (4.5,2.7);
		\node[label={[align=center]preparation\\circuit}] at (2,1.5) {};
		\draw[fill=white, thick] (0.5,0.25) rectangle (1.5,-1.75) node[midway] {$Q_j\mt$};
		\draw[fill=white, thick] (2.5,0.25) rectangle (3.5,-1.75) node[midway] {$Q_j^{-1}$};
		\node at (0.2,-1.1) {$\vdots$};
		\node at (3.8,-1.1) {$\vdots$};
		\node at (2,-1.1) {$\vdots$};
		\end{scope}
		\end{tikzpicture}
		\caption{The quantum operator $U$ in the preparation circuit controlled by an entry $u_j$ of $u$, in binary representation with classical bits. We replace the classical control by first implementing the circuit $Q_j$, applying the controlled-$U$ operator, and implementing $Q_j$ in reverse.}\label{fig:dist}
	\end{figure}
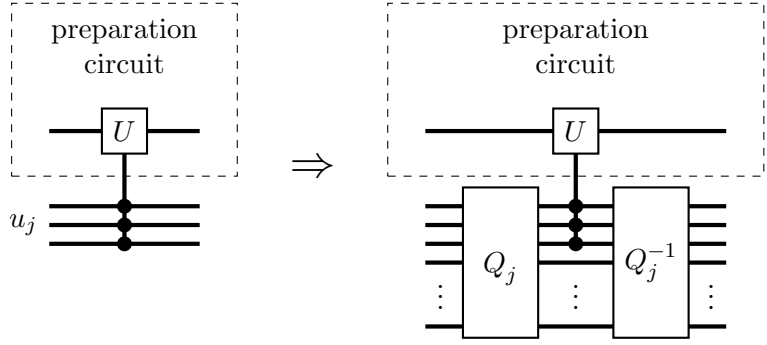
	
	It is clear that the entire circuit is with time $\poly(mT/\varepsilon)$ and space $O(\log(mT/\varepsilon))$. The error introduced by replacing each of th $r-1$ unitary operators is at most $(2m)^{-1}\varepsilon+2^{-\poly(mT/\varepsilon)}$, therefore the total error is at most $\varepsilon/2+(r-1)((2m)^{-1}\varepsilon+2^{-\poly(mT/\varepsilon)})<\varepsilon$. See Figure~\ref{fig:dist} for an illustration.
\end{proof}

The measurement $\{M_0,\ldots,M_{r-1}\}$ in \cref{thm:unitals} could be over any subset of the qubits. In particular, when it is a two-outcome measurement over one qubit, we have the following direct corollary which improves \cref{thm:unital}:

\begin{corollary}
	Given a unital quantum algorithm with time $T$ and space $S=\log m$ specified by the natural representations $K(\Phi_1),\ldots,K(\Phi_T\mt)\in\cc^{m^2\times m^2}$, and an error parameter $\varepsilon>0$, there is a unitary quantum circuit with time $\poly(mT/\varepsilon)$ and space $O(\log(mT/\varepsilon))$, such that if the original unital circuit outputs $0$ with probability $p$, then the unitary circuit outputs $0$ with probability $p+\alpha$, where $|\alpha|\leq \varepsilon$.
\end{corollary}

%% file: powering.tex
\section{Powering of Non-Contraction Matrices}\label{sect:powering}

In this section we extend the result of \cref{col:power} to matrices that may not necessarily be contractions. Notice that all the results in this section apply to iterative matrix multiplication (which asks about the product of $T$ different matrices instead of the same one as in powering) as well, using the same technique as at the start of the proof for \cref{thm:unital}. We first state the result for general square matrices, while the additive error can be exponentially large with respect to the spectral norm:
\begin{theorem}\label{thm:general}
	Given $m=2^S$, an arbitrary matrix $A\in\cc^{m\times m}$, a positive integer $T$, two unit vectors $v,w\in\cc^m$ and an error parameter $\varepsilon>0$, there is a unitary quantum circuit $W$ with time $\poly(mT/\varepsilon)$ and space $O(\log(mT/\varepsilon))$ such that with probability $1-2^{-\poly(mT/\varepsilon)}$, it outputs $w\ct A^Tv$ with additive error $\varepsilon\cdot\max(1,\|A\|^T)$.
\end{theorem}
\begin{proof}
	Ideally, we would like to apply the contraction powering algorithm on $A/\|A\|$ and multiply the result by $\|A\|^T$. However, the current best quantum algorithm for computing the spectral norm is \cite[Theorem 5.2]{ta2013inverting} which approximates $\|A\|$ with additive error $\varepsilon_1$ within time $\poly(m/\varepsilon_1)$ and space $O(\log(m/\varepsilon_1))$ and only works for contractions $A$. We use this algorithm to approximate $\|A\|$ for arbitrary $A$ with multiplicative error as follows\footnote{During the analysis we assume without loss of generality that $\|A\|\geq 1$, since otherwise it can always be relaxed to $1$ whenever necessary.}: First compute $\|A\|_\mathrm{F}$ in $O(\log m+\log \|A\|_F)=O(\log(m\|A\|_2))$ space. Notice that $A/\|A\|_\mathrm{F}$ is a contraction since $\|A\|_\mathrm{F}\geq \|A\|$. Therefore, let $\sigma$ be the approximation of $\|A/\|A\|_\mathrm{F}\|$ with additive error $\varepsilon_1$ by \cite{ta2013inverting}, then $\sigma\|A\|_F$ approximates $\|A\|$ since
	\[\Big|\sigma\|A\|_F-\|A\|\Big|=\|A\|_F\cdot\Big|\sigma-\big\|A/\|A\|_F\big\|\Big|\leq \sqrt{m}\varepsilon_1\|A\|.\]
	Let $\varepsilon_1=(3T\sqrt{m})^{-1}$, and let $\alpha=(1-\sqrt{m}\varepsilon_1)^{-1}\sigma\|A\|_F$. Then
	\[\|A\|\leq \alpha\leq \frac{1+\sqrt{m}\varepsilon_1}{1-\sqrt{m}\varepsilon_1}\|A\|\leq (1+T^{-1})\|A\|.\]
	Now let $\widetilde{A}=\alpha^{-1}A$ so that $\widetilde{A}$ is always a contraction. Applying the contraction powering algorithm in \cref{col:power} on $\widetilde{A}$ with error $\varepsilon/3$ results in a unitary quantum circuit with time $\poly(mT/\varepsilon)$ and space $O(\log(mT/\varepsilon))$ which outputs $w\ct \widetilde{A}^Tv$ with additive error $\varepsilon/3$. Multiplying it by $\alpha^T$ gives the desired result, while the error is at most $\alpha^T\varepsilon/3\leq\varepsilon\cdot\|A\|^T$.
\end{proof}

%% file: learning.tex
\section{Classical Simulation of Quantum Learning}\label{sect:learn}

\subsection{Equivalence of Classical Simulation in Decision and Learning}
	\begin{theorem}\label{thm:learn1}
		If there are functions $t(\cdot,\cdot)$ and $s(\cdot,\cdot)$, such that every unitary quantum learning algorithm with time $T$ and space $S$ can be simulated classically with time $t(T,S)$ and space $s(T,S)$, then
		\[\mathsf{promiseBQ_UL}\subseteq\mathsf{promiseBPTISP}(t(\poly(n),O(\log n)),s(\poly(n),O(\log n))).\]
		Specifically, if every unitary quantum learning algorithm in time $T$ and space $S$ can be simulated classically with time $\poly(2^ST)$ and space $O(S+\log T)$, then $\mathsf{promiseBQ_UL}=\mathsf{promiseBPL}$.
	\end{theorem}
	\begin{proof}
		Suppose that we have a unitary quantum circuit with time $T(n)=\poly(n)$ and space $S(n)=O(\log n)$ that decides a partial function $f:X\rightarrow\{0,1\}$, where $X\subseteq\{0,1\}^n$. Let $\ud(x,i)$ be the unitary gate at the $i$-th step of the decision algorithm with input $x$, which can be constructed in time $\poly(n)$ and space $O(\log n)$.
		
		We can convert the quantum circuit to a learning algorithm as follows. Use $X$ directly as the sample space, while the samples are always constant $x$ for some fixed $x\in X$. The learning task is to distinguish between $x\in f^{-1}(0)$ or $x\in f^{-1}(1)$. Upon receiving the sample $x$, the learning algorithm simply applies the following unitary operator on $\cc^{2^{S(n)}}\otimes\cc^{T(n)}$:
		\[\ket{\psi}\ket{i}\rightarrow \big(\ud(x,i)\ket{\psi}\big)\ket{(i+1)\textrm{ mod } T(n)}\]
		so that after $T(n)$ steps it computes in the first register the same state as in the quantum circuit. Therefore it computes $f(x)$ and distinguishes between the two cases. Using the premises, we have a classical learning algorithm with time $t(\poly(n),O(\log n))$ and space $s(\poly(n),O(\log n))$ that accomplishes the same task. The classical learning algorithm can be viewed as a randomized decision algorithm that computes $f(x)$ by self-constructing the stochastic matrices in the same time and space.
		
		Alternatively, in the learning task we can have a potentially much smaller sample space $\Gamma=\{0,1\}\times[n]$, by viewing the learning problem as computing $f$ in the random-query model \cite{raz2020random}. For each $x\in X$, let distribution $D_x$ be the one that uniformly draws $i\in[n]$ and outputs $(x_i,i)$. Let $\mathcal{X}=\{D_x\mid x\in f^{-1}(0)\}$ and $\mathcal{Y}=\{D_x\mid x\in f^{-1}(1)\}$. Each $x_i$ can be retrieved within $O(n\log n)$ samples with high probability, therefore the quantum learning algorithm can compute each $\ud(x,i)$ with high probability in time $\poly(n)$ and space $O(\log n)$, and the rest of the proof is the same as above.
	\end{proof}
	
	\begin{theorem}\label{thm:learn2}
		If $\textsc{ContractionPowering}\in\mathsf{promiseBPTISP}(t(n),s(n))$, where $t(n)\geq\Omega(n)$ and $s(n)\geq\Omega(\log n)$, then every unital quantum learning algorithm with time $T$ and space $S$ can be simulated classically with time $t(\poly(2^ST))$ and space $s(\poly(2^ST))$.
	\end{theorem}
	\begin{proof}
		Suppose that we have a unital quantum learning algorithm with time $T$ and space $S=\log m$ that distinguishes between two distribution families $\mathcal{X}$ and $\mathcal{Y}$. Let $\ul(z)$ be the unital channel applied when receiving the sample $z$. With the sample distribution $D$, let $A=\e_{z\sim D}[K(\ul(z))]$. We note that $A$ is a contraction matrix of dimension $m^2\times m^2$ as every $K(\ul(z))$ is a contraction. Similar to proof of \cref{thm:unital}, the probability of the learning algorithm outputting $0$ is 
		\[\e_{z\sim D^T}\big[\vect(M_0)\ct K(\Phi_T\mt)\cdots K(\Phi_1\mt)\vect(\rho_0\mt)\big]=\vect(M_0)\ct A^T \vect(\rho_0\mt).\]
		What's different from \cref{thm:unital} is that here $A$ is not explicitly given. Instead, by \cref{lemma:ch}, each time an entry of $A$ is requested, it takes $\poly(mT)$ samples $z$ to approximate the entry to at most $O((m^{2.5}T)^{-1})$ error, so that the approximated matrix $\widetilde{A}$ differs from the actual matrix $A$ by at most $\|\widetilde{A}-A\|\leq O((\sqrt{m}T)^{-1})$. By \cref{lemma:error} it means that $\|\widetilde{A}^T-A^T\|\leq O(m^{-1/2})$. Therefore applying the contraction powering algorithm on $\widetilde{A}$ gives a classical learning algorithm that distinguishes $\mathcal{X}$ and $\mathcal{Y}$ in time $t(\poly(mT))$ and space $s(\poly(mT))$.
		
		The above scheme has two problems. First, a fixed matrix $\widetilde{A}$ cannot be directly stored, and if every time the same entry is requested, the entry is approximated as the average of a different batch of samples, it may result in different requested values for the same entry (even though the difference is small with high probability), similar to the problem in \cref{lemma:power}. However, unlike the case in \cref{lemma:power}, here the classical contraction powering algorithm is not explicitly given, and may not be robust against changing inputs.
		
		The solution to this problem is the \textit{shift and truncate} method by Saks and Zhou\cite{saks1999bphspace}, which has found numerous applications in space-bounded algorithms \cite{ta2013inverting} and derandomizations \cite{cai2006time, grossman2019reproducibility}. Concretely, let $P=t(\poly(mT))$ be the largest number of possible requests to entries of $A$ in the contraction powering algorithm, and take a uniform random number $\zeta\in[8L]$. For simplicity let $L=12\sqrt{2m}T$ and $N=24m^{2.5}T$. When the entry $A_{jk}$ is requested, the algorithm takes $t(\poly(mT))$ samples $z_i$ and calculate the average value $a$ of the $(j,k)$-entries of $K(\ul(z_i))$, so that $|a-A_{jk}|<\frac{1}{8NP}$ with probability at least $1-2^{-P}$. The value fed back for the request is 
		\[\widetilde{A}_{jk}=\frac{1}{N}\left\lfloor N\cdot\mathrm{Re}(a)+\frac{\zeta}{8P}\right\rfloor
		+\frac{\mathrm{i}}{N}\left\lfloor N\cdot\mathrm{Im}(a)+\frac{\zeta}{8P}\right\rfloor.\]
		We claim that with high probability, this value coincides with the fixed value
		\[\frac{1}{N}\left\lfloor N\cdot\mathrm{Re}(A_{jk})+\frac{\zeta}{8P}\right\rfloor
		+\frac{\mathrm{i}}{N}\left\lfloor N\cdot\mathrm{Im}(A_{jk})+\frac{\zeta}{8P}\right\rfloor.\]
		For the real part, as $|N\cdot\mathrm{Re}(a)-N\cdot\mathrm{Re}(A_{jk})|<\frac{1}{8P}$, there is at most one possibility for $\zeta$ such that $\left\lfloor N\cdot\mathrm{Re}(a)+\frac{\zeta}{8P}\right\rfloor\neq \left\lfloor N\cdot\mathrm{Re}(A_{jk})+\frac{\zeta}{8P}\right\rfloor$, which is of probability $\frac{1}{8P}$, and the same holds for the imaginary part. By the union bound on the bad events during all $L$ requests, with probability 
		\[1-\left(2^{-P}+\frac{1}{4P}\right)P\geq \frac{2}{3}\]
		for every $(j,k)$ the value $\widetilde{A}_{jk}$ are always the same, and $|\widetilde{A}_{jk}-A_{jk}|\leq \frac{\sqrt{2}}{N}=\frac{1}{m^2L}$, so $\|\widetilde{A}-A\|\leq L^{-1}$.
		
		The second problem is that because of the approximation error, $\widetilde{A}$ might not be a contraction matrix. This is easily fixed by using the matrix $\widetilde{A}'=\frac{L}{L+1}\cdot\widetilde{A}$
		as the input. Since $\|\widetilde{A}-A\|\leq L^{-1}$ with probability $2/3$, it is implied that
		\[\|\widetilde{A}'\|=\frac{L}{L+1}\cdot\|\widetilde{A}\|
		\leq \frac{L}{L+1}\cdot\left(1+L^{-1}\right)=1,\]
		\[\|\widetilde{A}'-A\|\leq \|\widetilde{A}-A\|+\frac{1}{L+1}\|\widetilde{A}\|\leq\frac{2}{L}.\]
		Since $\|\vect(M_0)\|_2=\sqrt{m/2}$, $\|\vect(\rho_0\mt)\|_2=1$, in this case we have (by \cref{lemma:error})
		\[\left|\vect(M_0)\ct(\widetilde{A}'^T-A^T)\vect(\rho_0\mt)\right|\leq \frac{\sqrt{2m}T}{L}=\frac{1}{12}.\]
		
		Since the error of the original quantum learning algorithm can be amplified to $1/4$ so that $\vect(M_0)\ct A^T \vect(\rho_0\mt)$ is in $[0,1/4]$ or $[3/4,1]$, we conclude that with probability $5/6$, 
		\[\vect(M_0)\ct \widetilde{A}'^T \vect(\rho_0\mt)\in[0,1/3]\textrm{ or }[2/3,1]\]
		Therefore the two cases can be distinguished by the classical contraction powering algorithms on $\widetilde{A}'$, and it can be repeated for constant rounds so that the total error rate is brought down to $1/3$.
	\end{proof}

	\begin{corollary}
		If $\textsc{ContractionPowering}\in\mathsf{promiseBPL}$, then every unital quantum learning algorithm with time $T$ and space $S$ can be simulated classically with time $\poly(2^ST)$ and space $O(S+\log T)$.
	\end{corollary}

	Since by \cref{thm:power} we already know $\textsc{ContractionPowering}\in\mathsf{promiseBQ_UL}$, combined with \cref{thm:learn1}, we get the equivalence between efficient simulations of decision problems and learning problems:
	\begin{theorem}\label{thm:lequiv}
		Every (unital) quantum learning algorithm with time $T$ and space $S$ can be simulated classically with time $\poly(2^ST)$ and space $O(S+\log T)$, if and only if $\mathsf{promiseBQ_UL}=\mathsf{promiseBPL}$.
	\end{theorem}
	Also, as we already know $\mathsf{promiseBQ_UL}\subseteq\mathsf{promiseL^2}$ \cite{watrous1999space}, we have the following unconditional result:
	\begin{corollary}\label{col:NC2}
		Every unital quantum learning algorithm with time $T$ and space $S$ can be simulated classically with time $2^{O(S^2+\log^2 T)}$ and space $O(S^2+\log^2 T)$.
	\end{corollary}

	\subsection{Classical Simulation when One Family is Singleton}
	
	\begin{theorem}\label{thm:singleton}
		If $\mathcal{Y}=\{Y\}$, then any quantum learning algorithm that distinguishes between $\mathcal{X}$ and $\mathcal{Y}$ within time $T$ and space $S$ can be simulated classically in time $\poly(2^ST)$ and space $O(S+\log T)$.
	\end{theorem}
	\begin{proof}
		Suppose that we have a unital quantum learning algorithm with time $T$ and space $S=\log m$ that distinguishes between $\mathcal{X}$ and $\{Y\}$. Let $\ul(z)$ be the unital channel applied when receiving the sample $z$. We already know from \cref{thm:learn2} that with the sample distribution $D$, the output probability is $\vect(M_0)\ct A^T \vect(\rho_0\mt)$ where $A=\e_{z\sim D}[K(\ul(z))]$ is a contraction matrix of dimension $m^2\times m^2$. Since the distribution $Y$ is fixed, the corresponding matrix $B=\e_{z\sim Y}[K(\ul(z))]$ is also fixed. Now for any $D\in\mathcal{X}$, we know
		\[\left|\vect(M_0)\ct (A^T-B^T) \vect(\rho_0\mt)\right|\geq 1/3,\]
		and thus by \cref{lemma:error}
		\[\|A-B\|_\mathrm{F}\geq\|A-B\|\geq \frac{1}{T}\|A^T-B^T\|\geq \frac{1}{3T}\sqrt{\frac{2}{m}}.\]
		which means there must exist $i,j\in [m^2]$ such that $|A_{i,j}-B_{i,j}|\geq \frac{2}{9}m^{-5}T^{-2}$.
		
		The classical simulation algorithm iterates over all $i,j\in [m^2]$. For each choice of $i,j$, the algorithm approximates $\e_z[K(\ul(z))_{i,j}]$, compares it to $B_{i,j}$, and claims the samples are drawn from a distribution in $\mathcal{X}$ if there exists $i,j$ such that
		\[\left|\e_z[K(\ul(z))_{i,j}]-B_{i,j}\right|\geq \frac{1}{9m^5T^2}.\]
		
		Since each entry of $K(\ul(z))$ can be computed in time $O(T)$ and space $O(S)$ and has magnitude at most $1$, \cref{lemma:ch} asserts that $\poly(mT)$ samples are enough for accuracy $(9m^5T^2)^{-1}$ with probability $2/3$, in which case the algorithm correctly distinguish $\mathcal{X}$ and $\{Y\}$.
	\end{proof}

%% file: appendix.tex
\appendix 

\section{Proof for \cref{thm:perm}}\label{app:perm}

	Given the permutation $\sigma\in S_m$, we define a mapping $\tau:\{1,\ldots,m\}\rightarrow\{1,\ldots,m\}$ as follows:
	\[\tau(i)=\textrm{First }\sigma^k(i)\textrm{ in } \{\sigma(i),\sigma(\sigma(i)),\ldots,\sigma^k(i),\ldots\}\textrm{ such that }\sigma^k(i)\geq i.\]
	We prove the following claim which decomposes the permutation into transpositions:
	\begin{claim}
	 $\sigma=(1\  \tau(1))(2\  \tau(2))\cdots(m\  \tau(m))$.
	\end{claim}
	\begin{proof}
		Let $\sigma_i=(1\  \tau(1))(2\  \tau(2))\cdots(i\  \tau(i))$. The goal is to prove $\sigma_m=\sigma$, and it suffices to prove that $\sigma_i(i)=\sigma(i)$ for every $i$, as the transposition $(j\ \tau(j))$ is irrelevant to $i$ whenever $j>i$. We discuss the following three cases:
		
		\begin{itemize}
			
		\item If $\sigma(i)=i$, then $i$ only appears in the transposition $(i\ \tau(i))$, so $\sigma_i(i)=i$.
		
		\item If $\sigma(i)>i$, then $\tau(i)=\sigma(i)$. In this case, for every $j<i$, we have $\tau(j)\neq\tau(i)$, since otherwise if $\tau(j)=\sigma^k(j)$ then $i=\sigma^{k-1}(j)>j$ which contradicts the definition of $\tau$. Therefore $\tau(i)$ is not affected by $\sigma_{i-1}$, and thus
		\[\sigma_i(i)=\sigma_{i-1}\circ\tau(i)=\tau(i)=\sigma(i).\]
		
		\item If $\sigma(i)<i$, we define a list $k_0>k_1>\cdots>k_r=1$ as follows.
		\begin{itemize}
			\item[-] Let $k_0$ be the $k\in[m]$ such that $\tau(i)=\sigma^k(i)$.
			\item[-] If $k_t$ is already defined, let 
			\[k_{t+1}=\arg\max_{0<k<k_t}\sigma^k(i).\]
			\item[-] Repeat until the list reaches $k_r=1$.
		\end{itemize}
	 	For the sake of simplicity in notations, let $j_t=\sigma^{k_t}(i)$, so by definition we have
	 	\[\tau(i)=j_0\geq i>j_1>\cdots>j_r=\sigma(i).\]
	 	Since $j_{t+1}$ is the largest element in $\{\sigma^k(i)\mid k_{t+1}\leq k<k_t\}=\{\sigma^k(j_{t+1})\mid 0\leq k<k_t-k_{t+1}\}$, we also know that $\tau(j_{t+1})=\sigma^{k_t-k_{t+1}}(j_{t+1})=j_t$.
	 	
	 	Finally, for every $j$ with $j_1<j<i$, we have $\tau(j)\neq j_0$, since otherwise if $\tau(j)=\sigma^k(j)=j_0$ for some $k\in[m]$, then $j=\sigma^{k_0-k}(i)$, so it must hold $k>k_0$ which means $i=\sigma^{k-k_0}(j)>j$ which contradicts to the definition of $\tau$. For similar reasons, for every $t=1,2,\ldots,r$ and $j$ with $j_{t+1}<j<j_t$, we have $\tau(j)\neq j_t$ (let $j_{r+1}=0$). Combining the above we have
	 	\begin{align*}
	 		& \sigma_i(i) = \sigma_{i-1}\circ\tau(i) = \sigma_{i-1}(j_0) \\
	 		=\ & \big[\sigma_{j_1}\circ (j_1+1\ \ \tau(j_1+1))\cdots(i-1\ \  \tau(i-1))\big](j_0) \\
	 		=\ & \sigma_{j_1}(j_0) =\sigma_{j_1}\circ\tau(j_1) = \sigma_{j_1-1}(j_1)\\
	 		=\ & \big[\sigma_{j_2}\circ (j_2+1\ \ \tau(j_2+1))\cdots(j_1-1\ \  \tau(j_1-1))\big](j_1) \\
	 		=\ & \sigma_{j_2}(j_1) = \cdots = \sigma_{j_r-1}(j_r)\\
	 		=\ & \big[(1\ \ \tau(1))(2\ \ \tau(2))\cdots(j_r-1\ \ \tau(j_r-1))\big](j_r) = j_r=\sigma(i).\qedhere
	 	\end{align*}
		\end{itemize}	
	\end{proof}
	Now notice that each $\tau(i)$ can be computed in space $O(\log m)$. Each transposition is a two-level unitary which can be implemented up to error $\varepsilon/m$ within time $\poly(m/\varepsilon)$ and space $O(\log(m/\varepsilon))$, via \cite[Claim 4.2]{ta2013inverting}. By \cref{lemma:error} the composition approximates $\sigma$ with error at most $\varepsilon$.